\documentclass[12pt,a4paper]{article}

\usepackage{amsmath,amsthm,amssymb,amscd,a4wide}

\usepackage{dsfont}
\usepackage{mathrsfs}
\usepackage{setspace}
\usepackage{hyperref}

\newcommand{\ud}{\mathrm{d}}

\newcommand{\cH}{{\mathcal H}}

\DeclareMathOperator{\tr}{Tr}

\newcommand{\eins}{\mathds{1}}

\numberwithin{equation}{section}

\newtheorem{theorem}{Theorem}[section]
\newtheorem{lemma}[theorem]{Lemma}
\newtheorem{prop}[theorem]{Proposition}
\newtheorem{cor}[theorem]{Corollary}

\theoremstyle{definition}

\newtheorem{rem}[theorem]{Remark}

\numberwithin{equation}{section}

\begin{document}

\thispagestyle{empty}

\vspace*{1cm}

\begin{center}

{\LARGE\bf On the impact of surface defects on a condensate of electron pairs in a quantum wire} \\

\vspace*{2cm}

{\large Joachim Kerner \footnote{E-mail address: {\tt Joachim.Kerner@fernuni-hagen.de}} }%

\vspace*{5mm}

Department of Mathematics and Computer Science\\
FernUniversit\"{a}t in Hagen\\
58084 Hagen\\
Germany\\

\end{center},

\vfill
	
	\begin{abstract}In this paper we are interested in understanding the impact of surface defects on a condensate of electron pairs in a quantum wire. Based on previous results we establish a simple mathematical model in order to account for such surface effects. For a system of non-interacting pairs, we will prove the destruction of the condensate in the bulk. Finally, taking repulsive interactions between the pairs into account, we will show that the condensate is recovered for pair densities larger than a critical one given the number of the surface defects is not too large.
	\end{abstract}

	
	
	\section{Introduction}
	In this paper we are interested in establishing a mathematical model which allows us to understand the impact of surface defects on a (Bose-Einstein) condensate of electron pairs in a simple quantum wire, namely the half-line $\mathbb{R}_+=[0,\infty)$. It is motivated by the seminal work of Bardeen, Schrieffer and Cooper~\cite{CooperBoundElectron,BCSI} which demonstrated that the superconducting phase in (type-I) superconductors results from a coherent behaviour of pairs of electrons (Cooper pairs) similar to the one occurring in Bose gases (Bose-Einstein condensation) \cite{MR04}. Discovered first by Onnes, the most striking experimental feature of superconductors is a vanishing of the electrical resistance below some critical temperature \cite{Onnes1991}. 
	
	In a superconductor, the formation of a Cooper pair is a result of the interaction of two electrons with the lattice constituting the solid (electron-phonon-electron interaction). Due to the negative binding energy of each formed pair, the many-particle ground state of the superconductor (which itself is formed of pairs only) is separated from the excited states by a finite energy gap $\Delta > 0$ \cite{BCSI,MR04}. This energy gap, on the other hand, is one of the most important features that distinguishes the superconducting from the normal conducting phase. It is therefore not surprising that explaining the formation of such a gap was the main objective of Coopers ground-breaking work \cite{CooperBoundElectron}.
	
	Superconductivity as described above is a bulk phenomenon. However, in solid state physics it has long become clear that surface effects play an important role in various situations \cite{GennesBoundaryEffects} and, in particular, affect the superconducting behaviour of metals \cite{BarrettSurfaceEffects,TransportVinokur,KashiwayaSurfaceEffects,Belevtsov2003,TheOxfordHandbookSmallSuperconductors}. Hence, starting from~\cite{KernerElectronPairs} (see also~\cite{KernerJMPInteractingPairs}) where rigorous results regarding a condensation of electron pairs in a simple quantum wire were obtained, we will construct a simple mathematical model that allows to incorporate surface defects which are small compared to the bulk. After setting up the model we will investigate condensation of pairs in the bulk: In a first result we show that no (eigen-) state in the bulk remains macroscopically occupied after taking the surface defects into account. Hence, from a physical point of view, it becomes favourable for the pairs to accumulate in the surface defects. In a second step we then introduce repulsive interactions between the pairs and establish the existence of different regimes, one in which condensation in the bulk prevails and one in which it does not. Most importantly, given the number of the surface defects is not too large, there exists a critical pair density such that the pairs condense in the bulk for densities larger than this critical one.
	
	\section{Formulation of the model}
	We consider the quantum wire which is modelled by the half-line $\mathbb{R}_+=[0,\infty)$. On this quantum wire we place, as in~\cite{KernerElectronPairs}, a system of two interacting electrons (with same spin) whose Hamiltonian shall formally be given by
	\begin{equation}\label{HamiltonianPair}
	H_p=-\frac{\partial^2}{\partial x^2}-\frac{\partial^2}{\partial y^2}+v_{b}(|x-y|)
	\end{equation}
	with a binding-potential $v_{b}:\mathbb{R}_+ \rightarrow \overline{\mathbb{R}}_+$ defined as 
	\begin{equation}
	v_{b}(x):=\begin{cases}
	0 \quad \text{if} \quad 0\leq x \leq d\ , \\
	\infty \quad \text{else}\ .
	\end{cases}
	\end{equation}
	Due to the binding potential, the two electrons form a pair whose spatial extension is characterised by the parameter $d > 0$. We refer to~\cite{KernerElectronPairs} where a mathematically rigorous realisation of \eqref{HamiltonianPair} was obtained via the construction of a suitable quadratic form on 
	\begin{equation*}
	L^2_{a}(\Omega):=\{\varphi \in L^2(\Omega)|\ \varphi(x,y)=-\varphi(y,x)\}
	\end{equation*}
	with $\Omega:=\{(x,y)\in \mathbb{R}^2_+|\ |x-y|\leq d  \}$ being the two-particle configuration space.
	
	Now, in order to incorporate (localised) surface effects we extend our Hilbert space. More explicitly, we shall be working on the direct sum
	\begin{equation}\label{OnePairHS}
	\cH=L^2_{a}(\Omega) \oplus \ell^2(\mathbb{N})
	\end{equation}
	which means that we couple the (continuous) quantum wire to a discrete graph which is supposed to model surface defects. From a physical point of view this seems reasonable in a regime where the surface defects are relatively small compared to the bulk.
	
	%
	%
	
	Furthermore, the Hamiltonian of a free pair (meaning without surface-bulk interactions) shall be given by
	\begin{equation}\label{FreeOperator}
	H_0=H_p \oplus \mathcal{L}(\gamma)\ ,
	\end{equation}
	$\mathcal{L}(\gamma)$ being the (weighted) graph Laplacian, i.e., $f\in \ell^2(\mathbb{N})$,
	\begin{equation}
	(\mathcal{L}(\gamma)f)(n):=\sum_{m}\gamma_{nm}\left(f(m)-f(n)\right)
	\end{equation}
	where $\gamma:=(\gamma_{n,m}  \in \mathbb{R}_+)=\gamma^T$ is the associated edge weight matrix~\cite{ChungSpectralGraphTheory}. Since our graph is actually assumed to be a path graph (or chain graph), one sets $\gamma_{mn}=\delta_{|n-m|,1}e_n$ with $(e_n)_{n \in \mathbb{N}} \subset \mathbb{R}_+$. 
	%
	%
	%
	
	%
	\section{On the condensate in the bulk without surface pair interactions}
	In order to study the effect of the surface defects on a condensate of electron pairs we shall investigate the condensation phenomenon similar to~\cite{KernerElectronPairs}. Of course, given one wishes to describe the dynamics of a pair one would like to add a non-diagonal interaction term to~\eqref{FreeOperator} which describes the coupling between the bulk and the surface. However, since we are interested in quantum statistical properties only, we will simplify the discussion in this paper by modelling the interaction as a diagonal operator. The coupling between the surface and the bulk is then realised through the ``heat bath'' \cite{RuelleSM}. More explicitly, we consider the one-pair operator 
	\begin{equation}\label{ReferenceOp}
	H_{\alpha}(\gamma):=H_p \oplus (\mathcal{L}(\gamma)-\alpha \eins)
	\end{equation}
	with $\alpha \geq 0$ some constant characterising the ``surface tension''. Now, since we want to investigate Bose-Einstein condensation of pairs one has to employ a suitable thermodynamic limit \cite{RuelleSM}. For this, the half-line $\mathbb{R}_+$ is replaced by the interval $[0,L]$, $L > 0$, and one considers the restriction 
	\begin{equation}
	H^{L}_{\alpha}(\gamma):=H_p|_{L^2(\Omega_L)} \oplus (\mathcal{L}(\gamma)-\alpha \eins)|_{\mathbb{C}^{n(L)}}
	\end{equation}
	of \eqref{ReferenceOp} defined on 
	\begin{equation}
	\cH_L=L^2(\Omega_L) \oplus \mathbb{C}^{n(L)}
	\end{equation}
	where $\Omega_L:=\{(x,y) \in \Omega|\ 0\leq x,y \leq L \}$ and $n(L) \in \mathbb{N}$ refers to the number of surface defects up to length $L$ of the wire. 
	
	Since $H^{L}_{\alpha}(\gamma)$ is a direct sum of two operators, one has $\sigma(H^{L}_{\alpha}(\gamma))=\sigma(H_p|_{L^2(\Omega_L)}) \cup \sigma((\mathcal{L}(\gamma)-\alpha \eins)|_{\mathbb{C}^{n(L)}})$. As a consequence, $H^{L}_{\alpha}(\gamma)$ has purely discrete spectrum (see~\cite{KernerElectronPairs} for a discussion of $H_p|_{L^2(\Omega_L)}$). In the following, the eigenvalues of $H_p^{L}:=H_p|_{L^2(\Omega_L)}$ shall be denoted by $E_n(L)$ and the corresponding eigenfunctions by $\varphi_n$, $n \in \mathbb{N}_0$. Similarly, the eigenvalues of $(\mathcal{L}(\gamma)-\alpha \eins)|_{\mathbb{C}^{n(L)}}$ by $\lambda_j(L)$ and the associated eigenfunctions by $f_j$, $j=1,...,n(L)$. In both cases, the eigenvalues are counted with multiplicity. 
	%
	%
	%
	
	Now, as a first result we establish the following.
	\begin{prop}\label{PropInfimumSpectrum} For all sequences of edge weights $(e_n)_{n \in \mathbb{N}} \subset \mathbb{R}_+$ one has 
		\begin{equation}
		\inf \sigma(H^{L}_{\alpha}(\gamma))=-\alpha\ .
		\end{equation}
		Furthermore, $E_0(L) \geq E_0:=\frac{2\pi^2}{d^2}$ and 
		\begin{equation}
		\lim_{L \rightarrow \infty}E_0(L)=E_0\ .
		\end{equation}
	\end{prop}
	\begin{proof} The first equation follows directly from the fact that zero is the lowest eigenvalue to the discrete Laplacian associated with the constant eigenfunction $(1,1,1,...,1)^{T} \in \mathbb{C}^{n(L)}$. The second part of the statement was proved in [Lemma~3.1,~\cite{KernerElectronPairs}].
	\end{proof}
	In order to investigate condensation of pairs we will work, as customary in statistical mechanics, in the grand canonical ensemble~\cite{VerbeureBook,RuelleSM}. The associated Gibbs state is
	\begin{equation}\label{GibbsState}
	\omega^L_{\beta,\mu_L}(\ \cdot\ ):=\frac{\tr_{\mathcal{F}_b}(e^{-\beta (\Gamma(H^{L}_{\alpha}(\gamma))-\mu_L \mathrm{N})}[\ \cdot\ ])}{Z(\beta,\mu_L)}\ ,
	\end{equation}
	where $\beta=\frac{1}{T} \in (0,\infty)$ is the inverse temperature, $\mu_L \in (-\infty,\mu_{max}(L))$ the chemical potential (with $\mu_{max}(L)$ specified later) and $Z(\beta,\mu_L)=\tr_{\mathcal{F}_b}(e^{-\beta (\Gamma(H^{L}_{\alpha})-\mu_L \mathrm{N})})$ the partition function. Furthermore, $\mathcal{F}_b$ is the bosonic Fock space over $\mathcal{H}_L$, 
	\begin{equation}
	\mathrm{N}=\sum_{j=1}^{n(L)}a^{\ast}_ja_j+\sum_{n=0}^{\infty}a^{\ast}_na_n
	\end{equation}
	the number operator and 
	\begin{equation}\label{HamiltonianFreeSecondQS}
	\Gamma(H^{L}_\alpha(\gamma))=\sum_{j=1}^{n(L)}(\lambda_j(L)-\alpha)a^{\ast}_ja_j+\sum_{n=0}^{\infty}E_n(L)a^{\ast}_na_n
	\end{equation}
	the second quantisation of $H^{L}_\alpha(\gamma)$, see~\cite{MR04,BEH08} for more details. Note here that $\{a^{\ast}_j,a_j\}$ are the creation and annihilation operators corresponding to the states $\{0 \oplus f_j\}_{j=1}^{n(L)}$ and $\{a^{\ast}_n,a_n\}$ the ones corresponding to the states $\{\varphi_n \oplus 0\}_{n=0}^{\infty}$. 
	%
	%
	
	Most importantly, in the grand-canonical ensemble there is an explicit formula for the number of pairs occupying a given eigenstate~\cite{RuelleSM}: for every state $\varphi_n \oplus 0$ with associated number operator $n_{\varphi_n}:=a^{\ast}_na_n$, the number of pairs occupying this state is 
	\begin{equation}\label{FormulaParticleDensityII}
	\omega^L_{\beta,\mu_L}(n_{\varphi_n})=\frac{1}{e^{\beta(E_n(L)-\mu_L)}-1}\ .
	\end{equation}
	An equivalent formula applies to any element of the form $0 \oplus f_j$, setting  $n_{f_j}:=a^{\ast}_ja_j$.
	
	For the Hamiltonian~\eqref{HamiltonianFreeSecondQS}, the thermodynamic limit shall then be realised as the limit $L \rightarrow \infty$ such that 
	\begin{equation}\label{SequenceCPEinsteinsArgument}
	\rho=\frac{1}{L}\left(\sum_{j=1}^{n(L)} \omega^L_{\beta,\mu_L}(n_{f_j})+\sum_{n=0}^{\infty}  \omega^L_{\beta,\mu_L}(n_{\varphi_n})\right)
	\end{equation}
	holds for all values of $L$ with $\mu_L$ denoting the sequence of the chemical potentials and $\rho > 0$ the pair density. Furthermore, we say that a bulk state $\varphi_n \oplus 0$, $n \in \mathbb{N}_0$, is macroscopically occupied in the thermodynamic limit if 
	\begin{equation}\label{EquationMacroscopicOccupation}
	\limsup_{L \rightarrow \infty}\frac{1}{e^{\beta(E_n(L)-\mu_L)}-1} > 0
	\end{equation}
	holds.

	As shown in~[Theorem~3.3,~\cite{KernerElectronPairs}], the state $\varphi_0 $ is macroscopically occupied in the thermodynamic limit given the underlying Hilbert space is $L^2(\Omega_L)$ only. In contrast to this, we obtain the following result when working on $\cH_L$, i.e., when including the surface defects.
	%
	
	\begin{theorem}[Destruction of the condensate in the bulk I]\label{TheoremDSNOn} Assume that $H^{L}_{\alpha}(\gamma)$ is given with an arbitrary sequence of edge weigths $(e_n)_{n \in \mathbb{N}} \subset \mathbb{R}_+$. Then, for the associated Gibbs state and all pair densities $\rho > 0$, no bulk state $\varphi_n \oplus 0$, $n \in \mathbb{N}_0$, is macroscopically occupied in the thermodynamic limit. Actually, one has 
		\begin{equation*}
		\lim_{L \rightarrow \infty}\frac{1}{L}\frac{1}{e^{\beta(E_n(L)-\mu_L)}-1}=0\ , \quad \forall n \in \mathbb{N}_0\ .
		\end{equation*}
	\end{theorem}
	\begin{proof} Since $\mu_{max}(L)=\inf \sigma(H^{L}_{\alpha}(\gamma))$ in the non-interacting case \cite{RuelleSM}, Proposition~\ref{PropInfimumSpectrum} implies that $\mu_L \in (-\infty,-\alpha)$. The result then readily follows from \eqref{EquationMacroscopicOccupation} taking Proposition~\ref{PropInfimumSpectrum} into account.
	\end{proof}
	\section{On the condensate in the bulk in the presence of surface pair interactions}
	In the previous section we have seen, by Theorem~\ref{TheoremDSNOn}, that the condensate of electron pairs in the bulk is destroyed through the presence of surface defects. However, since the defects are imagined relatively small when compared to the bulk, (repulsive) interactions between the pairs should be taken into account for large pair surface densities. 
	
	In order to account for those interactions, we pursue a (quasi) mean-field approach. More explicitly, the first term on the right-hand side of \eqref{HamiltonianFreeSecondQS} (i.e., the free Hamiltonian associated with the discrete graph) shall be replaced by
	\begin{equation}\label{MeanFieldHamiltonian}
	\sum_{j=1}^{n(L)}\left(\lambda_{j}(L)-\alpha+\lambda \rho_{s}(\mu_L,L) \right)a^{\ast}_ja_j:=h_L(\alpha,\lambda)\ ,
	\end{equation}
	where $\rho_{s}(\mu_L,L) \geq 0$ is the pair density in the surface defects, i.e., on $\mathbb{C}^{n(L)}$; see eq.~\eqref{SurfaceDensityII} below. Furthermore, $\lambda > 0$ is the interaction strength associated with the repulsive interactions between the pairs. 
	%
	\begin{rem} Note that the interaction term in the standard mean-field approach is $\lambda \frac{\mathrm{N}^2}{V}$ with $V$ the associated volume, see~\cite{MichVerbeure,VerbeureBook}.
	\end{rem}
	Since we can write $h_L(\alpha,\lambda)$ as in~\eqref{MeanFieldHamiltonian} we conclude that  the eigenvalues $\lambda_j(L)$ are effectively only shifted by $\lambda\rho_{s}(\mu,L)-\alpha$. Accordingly, the problem thereby reduces to an effective non-interacting ``particle'' model and one has
	\begin{equation}\label{RelationChemicalPotential}
	\mu_{L} < \min\{\lambda\rho_{s}(\mu_L,L)-\alpha,E_0(L)\}
	\end{equation}
	for the sequence of chemical potentials $\mu_L$~\cite{RuelleSM}, taking into account that the lowest eigenvalue of the Laplacian is zero. In particular, $\mu_{max}(L)=\min\{\lambda\rho_{s}(\mu_L,L)-\alpha,E_0(L)\}$. Furthermore, $\mu_{L}$ and the surface pair density $\rho_{s}(\mu_L,L)$ shall be chosen in a way such that
	\begin{equation}\label{EquationI}
	\rho=\lim_{k \rightarrow \infty}\frac{1}{L_k}\left(\sum_{j=1}^{n(L_k)}\frac{1}{e^{\beta\left[\left(\lambda_{j}(L_k)-\alpha+\lambda \rho_{s}(\mu_{L_{k}},L_k) \right)-\mu_{L_k}\right]}-1}+\sum_{n=0}^{\infty}\frac{1}{e^{\beta\left(E_n(L_k) -\mu_{L_k}\right)}-1}\right)\ ,
	\end{equation}
	for a subsequence $\mu_{L_k}$ together with
	\begin{equation}\label{SurfaceDensityII}
	\rho_{s}(\mu_{L},L)=\frac{1}{n(L)}\sum_{j=1}^{n(L)}\frac{1}{e^{\beta\left[\left(\lambda_{j}(L)-\alpha+\lambda \rho_{s}(\mu_L,L) \right)-\mu_{L}\right]}-1}\ .
	\end{equation}
	%
	%
	%
	In the rest of the section we shall assume that $L/n(L)$ is bounded from above and that $\mu_{L_k}$ converges to a (possibly negative 
	infinite) limit value $\mu \leq E_0$ (from eq.~\eqref{EquationI} and eq.~\eqref{SurfaceDensityII} we indeed conclude that there are values and, in particular, arbitrarily large/small values of $\rho > 0$ for which such sequences exist). Note that, for notational simplicity we will, in the following, restrict ourselves to subsequences without further notice.
	\begin{theorem} Let $\mu_L \in (-\infty,E_0(L))$ be a corresponding sequence of chemical potentials with limit value $\mu \leq E_0$. Then 
		\begin{equation}\label{EquationTheorem}
		\lim_{L \rightarrow \infty}\left(\frac{n(L)}{L}\rho_s(\mu_L,L)+\rho_0(\mu_L,L)\right)=\rho-\frac{\sqrt{2}}{\pi}\sum_{n=1}^{\infty}\int_{0}^{\infty}\frac{1}{e^{\beta \frac{2\pi^2n^2}{d^2}}e^{\beta(x^2-\mu)}-1}\ \ud x\ ,
		\end{equation}
		where $\rho_0(\mu_L,L):=\omega^L_{\beta,\mu_L}(n_{\varphi_0})/L$.
	\end{theorem}
	\begin{proof} Starting from~\eqref{EquationI}, the statement follows directly from formula~(3.4) of~\cite{KernerElectronPairs} setting $\hbar=1$, $m_e=1/2$ and replacing $d$ by $d/\sqrt{2}$.
		
	\end{proof}
	Writing $\lim_{L\rightarrow \infty}\frac{L}{n(L)}=:\delta$ with $0\leq \delta < \infty$ then, setting $\rho_0:=\lim_{L \rightarrow \infty}\rho_0(\mu_L,L)$,
	\begin{equation}\label{EquationLimitDensity}\begin{split}
	\lim_{L\rightarrow \infty}\rho_s(\mu_L,L)& = \delta\left(\rho-\frac{\sqrt{2}}{\pi}\sum_{n=1}^{\infty}\int_{0}^{\infty}\frac{1}{e^{\beta \frac{2\pi^2n^2}{d^2}}e^{\beta(x^2-\mu)}-1}\ud x\right)-\delta \rho_0 \\
	&=:\tilde{\rho}(\mu,\delta)-\delta \rho_0 \ ,
	\end{split}
	\end{equation}
	for a limit value $\mu \leq E_0$. From a physical point of view it is also interesting to write 
	\begin{equation}
	\tilde{\rho}(\mu,\delta)=\delta(\rho-\rho_{exc})\ ,
	\end{equation}
	where $\rho_{exc}=\rho_{exc}(\beta,\mu)$ equals the second term within the brackets in eq.~\eqref{EquationLimitDensity}. Note that $\rho_{exc}$ is the density of pairs occupying all excited (eigen-)states (i.e., $\varphi_n \oplus 0$ with $n\geq 1$) in the bulk in the thermodynamic limit.
	\begin{lemma}[Destruction of the condensate in the bulk II]\label{ProofLemma}  Assume that $\lambda > 0$ and
		\begin{equation}\label{EquationConditionTheorem}
		\tilde{\rho}(\mu,\delta) < \frac{E_0+\alpha}{\nu\lambda}
		\end{equation}
		for some $\nu >1$ and $\mu$ the limit point of $\mu_L$. Then
		\begin{equation}\label{EquationLimitZero}
		\lim_{L \rightarrow \infty}\frac{1}{L}\frac{1}{e^{\beta(E_n(L)-\mu_L)}-1}=0\ , \quad \forall n \in \mathbb{N}_0\ .
		\end{equation}
	\end{lemma}
	\begin{proof} By relation~\eqref{RelationChemicalPotential} and the assumptions we conclude that
		\begin{equation}
		\mu <  E_0-\varepsilon\ 
		\end{equation}
		for the limit point of $\mu_L$ and for some constant $\varepsilon > 0$. Consequently, by Proposition~\ref{PropInfimumSpectrum} we get
		\begin{equation}
		\lim_{L \rightarrow \infty}\frac{1}{L}\frac{1}{e^{\beta(E_0(L)-\mu_L)}-1}=0\ , \quad \forall n \in \mathbb{N}_0\ ,
		\end{equation}
		which yields the statement since the bulk ground state is occupied the most.
	\end{proof}
	We immediately obtain the following corollary which is particularly interesting from a physical point of view.
	\begin{cor}\label{CorollarySurface} Assume that $\delta=0$. Then, for all values $\lambda > 0$, \eqref{EquationLimitZero} holds for the states $\varphi_n \oplus 0$, $n \in \mathbb{N}_0$. 
		
		More generally, if 
		\begin{equation}
		\delta \cdot \rho < \frac{E_0+\alpha}{\nu\lambda}
		\end{equation}
		for some $\nu > 1$ and $\lambda > 0$, then~\eqref{EquationLimitZero} holds for the states $\varphi_n \oplus 0$, $n \in \mathbb{N}_0$. 
	\end{cor}
	\begin{rem} Corollary~\ref{CorollarySurface} implies that the condensate of electron pairs in the bulk (which exists due to [Theorem~3.3,~\cite{KernerElectronPairs}] whenever no surface defects are present) is destroyed for arbitrarily large repulsive (quasi) mean-field interactions if the number of surface defects is large, i.e., of order larger than $L$. 
		
		In addition, Lemma~\ref{ProofLemma} implies that the condensate in the bulk is destroyed for arbitrarily large pair densities $\rho >0$ given the interaction strength $\lambda > 0$ is small enough or the surface tension $\alpha \geq 0$ large enough.
	\end{rem}
	Finally, we obtain the following result.
	\begin{theorem}[Reconstruction of the condensate] If $\delta,\lambda > 0$ then there is a critical pair density $\rho_{crit}=\rho_{crit}(\beta,\delta,\alpha,\lambda) > 0$ such that for all pair densities $\rho > \rho_{crit}$ one has
		\begin{equation*}
		\lim_{L \rightarrow \infty}\frac{1}{L}\frac{1}{e^{\beta(E_0(L)-\mu_L)}-1} > 0 \ .
		\end{equation*}
	\end{theorem}
	\begin{proof} Assume to the contrary that such a critical pair density doesn't exist. Then there exist arbitrarily large $\rho > 0$ for which $\rho_0=0$. 
		
		Now, in a first step pick such a $\rho$, for given values $\beta,\delta,\alpha,\lambda > 0$, so large that
		\begin{equation}\label{EQ1}
		\tilde{\rho}(E_0,\delta) > 2\frac{E_0+\alpha}{\lambda}\ .
		\end{equation}
		Using the same reasoning as in the proof of Lemma~\ref{ProofLemma} one then concludes that 
		\begin{equation}\label{EQ2}
		\mu \leq E_0
		\end{equation}
		for the limit point of the associated sequence of chemical potentials, taking Proposition~\ref{PropInfimumSpectrum} into account.
		
		In a second step we use~\eqref{EQ1} and~\eqref{EQ2} in~\eqref{SurfaceDensityII} to conclude, for such $\rho$, the existence of a constant $C > 0$ such that 
		\begin{equation}
		|\rho_{s}(\mu_L,L)| < C
		\end{equation}
		for all $L\geq L_0$, $L_0$ large enough. Finally, increasing $\rho$ even more then yields a contradiction with~\eqref{EquationLimitDensity} and consequently the statement.
	\end{proof}
	\begin{rem} We remark that the results in this section hold for an arbitrary sequence $(e_n)_{n \in \mathbb{N}} \subset \mathbb{R}_+$ of edge weights. 
	\end{rem}
	%






{\small
\bibliographystyle{amsalpha}
\bibliography{Literature}}

\end{document}